\documentclass[runningheads]{llncs}
\usepackage{amssymb}
\usepackage{amsmath}
\usepackage{comment}
\usepackage[disable]{todonotes}
\usepackage{inputenc}
\usepackage{shuffle}
\usepackage{multirow}
\usepackage{listings}
\usepackage{mathabx}
\usepackage{hyperref}

\usepackage{shuffle}

\usepackage{tikz}
\usetikzlibrary{positioning,shadows,arrows}
\usetikzlibrary{arrows,automata,positioning}
\usetikzlibrary{shapes,shapes.geometric,arrows,fit,calc,positioning,automata}

\usepackage{diagbox}
\usepackage{slashbox}
\usepackage{tikz}
\usetikzlibrary{matrix}

\newcommand{\perm}{\operatorname{perm}}
\newcommand{\stc}{\operatorname{sc}}

\usepackage{apxproof}
\theoremstyle{plain}
\newtheoremrep{theorem}{Theorem}
\newtheoremrep{proposition}[theorem]{Proposition}
\newtheoremrep{lemma}[theorem]{Lemma}
\newtheoremrep{claim}[theorem]{Claim}
\newtheoremrep{conjecture}[theorem]{Conjecture}
\newtheoremrep{corollary}[theorem]{Corollary}
\newtheoremrep{definition}[theorem]{Definition}

\newenvironment{claiminproof}[1]{\medskip\par\noindent\underline{Claim:}\space#1}{}
\newenvironment{claimproof}[1]{\begin{quote}\par\noindent\emph{Proof of the Claim:}\space#1}{[\emph{End, Proof of the Claim}]\end{quote}}

\DeclareMathOperator{\lcm}{lcm}

\DeclareFontFamily{U}{bigshuffle}{}
\DeclareFontShape{U}{bigshuffle}{m}{n}{
  <5-8> s*[1.7] shuffle7
  <8->  s*[1.7] shuffle10
}{}
\DeclareSymbolFont{BigShuffle}{U}{bigshuffle}{m}{n}
\DeclareMathSymbol\bigshuffle{\mathop}{BigShuffle}{"001}
\DeclareMathSymbol\bigcshuffle{\mathop}{BigShuffle}{"002}

\newcommand{\upwardclosure}[1]{\ensuremath{\mathop{\uparrow\!} #1}}

\begin{document}
%



\title{State Complexity Investigations on Commutative Languages -- The Upward and Downward Closure, Commutative Aperiodic and Commutative Group Languages} 

\titlerunning{State Complexity on Commutative Aperiodic and Group Languages}

%
%
\author{Stefan Hoffmann\orcidID{0000-0002-7866-075X}}
\authorrunning{S. Hoffmann}
%
\institute{Informatikwissenschaften, FB IV, 
  Universit\"at Trier,  Universitätsring 15, 54296~Trier, Germany, 
  \email{hoffmanns@informatik.uni-trier.de}}
\maketitle              
\begin{abstract}
 We investigate the state complexity of the upward and downward closure
 and interior operations on commutative regular languages.
 Then, we systematically study the state complexity of these operations
 and of the shuffle operation on commutative group languages
 and commutative aperiodic (or star-free) languages.
 
\end{abstract}

\section{Introduction}
\label{sec:introduction}

The state complexity, as used here, of a regular language $L$
is the minimal number of states needed in a complete deterministic automaton
recognizing~$L$. The state complexity of an operation
on regular languages is the greatest state complexity
of the result of this operation
as a function of the (maximal) state complexities of its arguments.

Investigating the state complexity of the result of a regularity-preserving operation on regular languages,
see~\cite{GaoMRY17} for a survey, was first initiated by Maslov in~\cite{Mas70} and systematically started by Yu, Zhuang \& Salomaa in~\cite{YuZhuangSalomaa1994}.

A language is called commutative, if
for any word in the language, every permutation of this word is also in the language.
The class of commutative automata, which recognize commutative regular languages, was introduced in~\cite{BrzozowskiS73}.


The shuffle operation has been introduced to understand the semantics of parallel programs~\cite{CamHab74,DBLP:conf/mfcs/Marzurkiewicz75,DBLP:journals/cl/Riddle79,Shaw78zbMATH03592960}.
The shuffle operation 
is regularity-preserving on all regular languages. The state complexity of the shuffle
operation in the general cases was investigated in~\cite{BrzozowskiJLRS16}
for complete deterministic automata and in~\cite{DBLP:journals/jalc/CampeanuSY02}
for incomplete deterministic automata. The bound $2^{nm-1} + 2^{(m-1)(n-1)}(2^{m-1}-1)(2^{n-1}-1)$ was obtained
in the former case, which is not known to be tight in case of complete automata, and the tight bound $2^{nm}-1$
in the latter case. 

A word is a (scattered) subsequence of another word, if it can be obtained from the latter word by deleting
letters. This gives a partial order, and the upward and downward closure and interior operations, denoted by $\mathop{\uparrow\!} U$,
$\mathop{\downarrow\!} U$,
$\mathop{\uptodownarrow\!} U$
and $\mathop{\downtouparrow\!} U$,
refer to this partial order. Languages that result from upward closure operation are also known as shuffle ideals.
The state complexity of these operations was investigated in~\cite{DBLP:journals/tcs/GruberHK07,DBLP:journals/fuin/GruberHK09,DBLP:journals/ita/Heam02,KarandikarNS16,DBLP:journals/fuin/Okhotin10}


In~\cite{Hoffmann2021NISextended,DBLP:conf/cai/Hoffmann19,DBLP:conf/dcfs/Hoffmann21,Hoffmann2021} the state complexity of these operations
was investigated in the case of commutative regular languages.
The results are summarized in Table~\ref{tab:sc_known_results}.


\begin{table}[ht]
    \centering 
    \begin{tabular}{|c|c|c|c|} 
    \hline
    Operation                                    &  Upper Bound        &  Lower Bound              & Reference \\ \hline 
    $\pi_{\Gamma}(U)$, $\Gamma \subseteq \Sigma$ & $n$                 &  $n$  &  \cite{Hoffmann2021NISextended,Hoffmann2021} \\
    $U \shuffle V$ & $\min\{(2nm)^{|\Sigma|},f(n,m)\}$ & $nm$ &  \cite{BrzozowskiJLRS16,Hoffmann2021NISextended,DBLP:conf/cai/Hoffmann19} \\
    $\uparrow\! U$                               & $\min\{n^{|\Sigma|},2^{n-2}+1\}$ & $\Omega\left( \left( \frac{n}{|\Sigma|} \right)^{|\Sigma|} \right)$ & Thm.~\ref{thm:sc_upper_bound_closure_interior} \& \cite{DBLP:journals/ita/Heam02,Hoffmann2021,KarandikarNS16} \\
    $\downarrow\! U$                             & $\min\{n^{|\Sigma|},2^{n-1}\}$ & $n$ &  Thm.~\ref{thm:sc_upper_bound_closure_interior} \&\cite{Hoffmann2021,KarandikarNS16} \\
    $\uptodownarrow\! U$                         & $\min\{n^{|\Sigma|},2^{n-2}+1\}$ & $\Omega\left( \left( \frac{n}{|\Sigma|} \right)^{|\Sigma|} \right)$ & Thm.~\ref{thm:sc_upper_bound_closure_interior} \& \cite{Hoffmann2021,KarandikarNS16} \\
    $\downtouparrow\! U$                         & $\min\{n^{|\Sigma|},2^{n-1}\}$ & $n$ & Thm.~\ref{thm:sc_upper_bound_closure_interior} \& \cite{Hoffmann2021,KarandikarNS16} \\
    $U \cup V$,$U \cap V$                        & $nm$               & sharp, for any $\Sigma$ & \cite{Hoffmann2021NISextended,Hoffmann2021} \\ \hline
    \end{tabular}
    \caption{State complexity results on commutative regular languages, where $n$
    and $m$ denote the state complexities of the input languages.
     Also, $f(n,m) =  2^{nm-1} + 2^{(m-1)(n-1)}(2^{m-1}-1)(2^{n-1}-1)$
     is the general bound for shuffle from~\cite{BrzozowskiJLRS16}.} 
    \label{tab:sc_known_results}
\end{table} 

\todo{upward und downward closure hier annoucnen und in obere tabelle eintragen.}

\todo[inline]{
 star-free paper Brzozowski
 schnitt vereinigung für binär scharf, aber 
 für mehr weiß ich nicht, ob diekonstruktin so geht. z.B. zustand (4,5) wenn man ein c hat, werden die alle zu einem zustand?
 oder

 das mit den alphabeten erstmal nicht großartig thematisieren.
 
 sprechweise quotient complexity for complete state complexity?
}

\todo{$|\Sigma|$ durch $k$ ersetzen, oder andersrum?}

\todo{referenzen einträge über mehrere zeilen aufteilen?}

\todo{sharpness fuer beide parameter vergleichen.}

\renewcommand{\arraystretch}{1.2}
\begin{table}[ht]
    \centering 
    \begin{tabular}{|c|c|c|c|c|>{\raggedright\arraybackslash}p{2.2cm}|} 
    \hline
     & \multicolumn{2}{c|}{Group Case} & \multicolumn{2}{c|}{Aperiodic Case} & \\ \hline
    Op.            &  $\le$   &  $\ge$  & $\le$  & $\ge$  & Reference \\ \hline  
    $\pi_{\Gamma}(U)$ & $n$   &  $n$    & $n$    & $n$    & \cite{Hoffmann2021} \\ \hline
    $U \shuffle V$        & $(nm)^{|\Sigma|}$         & $nm$ & $(n + m - 1)^{|\Sigma|}$ & $\left\{ \begin{array}{ll} \Omega\left( nm \right) & \mbox{if } |\Sigma| > 1 \\ 
    n + m - 1 & \mbox{if } |\Sigma| = 1 \end{array}\right.$ &   
    Thm.~\ref{thm:group:shuffle_upper_bound}, 
    Prop.~\ref{prop:group:shuffle_lower_bound}, \cite{BrzozowskiL12,Hoffmann2021NISextended,DBLP:conf/cai/Hoffmann19} \& Prop.~\ref{prop:aperiodic:shuffle_lower_bound} 
    \\ \hline
    $\uparrow\! U$        & $n^{|\Sigma|}$ & n & $\min\{n^{|\Sigma|},2^{n-2}+1\}$ & $\Omega\left( \left( \frac{n}{|\Sigma|} \right)^{|\Sigma|} \right)$ & Thm.~\ref{thm:sc_upper_bound_closure_interior}, Prop.~\ref{prop:group:upward_closure_lower_bound}, Prop.~\ref{prop:aperiodic:upward_lower_bound} \& \cite{KarandikarNS16} \\ \hline
    $\downarrow\! U$      & $1$ & $1$ & $\min\{n^{|\Sigma|},2^{n-1}\}$    & n &  Thm.~\ref{thm:sc_upper_bound_closure_interior}, Prop.~\ref{prop:group:downward_upper_bound}, Prop.~\ref{thm:aperiodic:interior_lower_bound} \& \cite{KarandikarNS16} \\ \hline
    $\uptodownarrow\! U$  & $n^{|\Sigma|}$ & $n$ &  $\min\{n^{|\Sigma|},2^{n-2}+1\}$ & $\Omega\left( \left( \frac{n}{|\Sigma|} \right)^{|\Sigma|} \right)$ & Eq.~\eqref{eqn:interior} \\ \hline
    $\downtouparrow\! U$  & $1$ & $1$                     & $\min\{n^{|\Sigma|},2^{n-1}\}$ & $n$ & Eq.~\eqref{eqn:interior} \\ \hline
    \parbox{1cm}{\centering $U \cap V$ \\ $U\cup V$}            & $nm$           & $nm$ &  
    $\left\{\begin{array}{ll}
     nm & |\Sigma| \ge 2 \\
     \max\{n,m\} & |\Sigma| = 1
    \end{array}\right.$ & $\left\{\begin{array}{ll}
     nm & |\Sigma| \ge 2 \\
     \max\{n,m\} & |\Sigma| = 1
    \end{array}\right.$  & Thm~\ref{thm:group:union_intersection} \& \cite{BrzozowskiL12}\\ \hline
    \end{tabular} 
    \caption{The state complexity results for various operations for input languages of state complexities $n$ and $m$. The upper bound ($\le$)
     and the best known lower bound ($\ge$) are indicated for the group 
     and the aperiodic commutative languages.
     Also, $f(n,m) =  2^{nm-1} + 2^{(m-1)(n-1)}(2^{m-1}-1)(2^{n-1}-1)$
     is the general bound for shuffle from~\cite{BrzozowskiJLRS16}. So, the bound for shuffle
     is actually the minimum of the bound stated and $f(n,m)$ as written in Table~\ref{tab:sc_known_results} (this is left out to save horizontal space). For the lower bound for projection,
     consider the group language $(a^n)^* \shuffle \Sigma\setminus\{a\}$
     and the aperiodic language $\{ a^{n-1} \} \shuffle \Sigma\setminus\{a\}$ and $\Gamma \subseteq \Sigma \setminus \{a\}$.}
    \label{tab:sc_results_here}
\end{table} 

A \emph{group language} is a language recognizable by an automaton where every letter induces a permutation
of the state set. The investigation of the state complexity of these languages was started
recently~\cite{DBLP:conf/dlt/HospodarM20}.

A \emph{star-free language} is a language which can be written with an extended regular expression, i.e.,
an expression involving concatenation, the Boolean operations and Kleene star,
without using the Kleene star~\cite{McNaughton67}. The class of star-free languages coincides with the class
of \emph{aperiodic languages}~\cite{Schutzenberger65a}, i.e., those languages recognizable by automata such that no subset of states
is permuted by a word. 


So, in this sense the aperiodic languages are as far away from the group languages
as possible. In~\cite{Hoffmann2021NISextended,DBLP:conf/cai/Hoffmann19} it has been shown that every commutative
and regular language can be decomposed into commutative aperiodic and commutative group languages 
in the following way.

\begin{theorem}[\cite{Hoffmann2021NISextended,DBLP:conf/cai/Hoffmann19}]
\label{thm::decomp}
 Suppose $L \subseteq \Sigma^*$ is commutative and regular. Then,\todo{in introduction als motivation?}
 $L$ is a finite union of languages of the form $U \shuffle V$,
 where $U$ is a commutative aperiodic language
 and $V$ is a commutative group language
 over a subalphabet\footnote{Over the whole alphabet $\Sigma$,
 these languages are precisely the languages recognizable by automata whose transition
 monoids are $0$-groups, i.e., groups with a zero element adjoined.} of $\Sigma$.
\end{theorem}

Here, we will investigate the state complexity of operations considered
in~\cite{Hoffmann2021NISextended,DBLP:conf/cai/Hoffmann19} for general commutative regular languages
for the commutative group and the commutative aperiodic languages separately. Additionally, we will investigate four new operations -- the upward
and downward closure and interior operations, denoted by $\mathop{\uparrow\!} U$,
$\mathop{\downarrow\!} U$,
$\mathop{\uptodownarrow\!} U$
and $\mathop{\downtouparrow\!} U$ -- for which we first state a bound on general
commutative regular languages and then also bounds for commutative aperiodic
and commutative group languages. See Table~\ref{tab:sc_known_results} and Table~\ref{tab:sc_results_here} 
for a summary of the results.

\section{Preliminaries}

In the present work,
we assume that $k \ge 0$ denotes our \emph{alphabet} size and $\Sigma = \{a_1, \ldots, a_k\}$.
We will also write $a,b,c$ for $a_1,a_2,a_3$ in case of $|\Sigma| \le 3$.
The set $\Sigma^{\ast}$ denotes
the set of all finite sequences, i.e., of all \emph{words}. The finite sequence of length zero,
or the \emph{empty word}, is denoted by $\varepsilon$. For a given word we denote by $|w|$
its \emph{length}, and for $a \in \Sigma$ by $|w|_a$ the \emph{number of occurrences of the symbol $a$}
in $w$. A \emph{language} is a subset of $\Sigma^*$.

The \emph{shuffle operation}, denoted by $\shuffle$, is defined by
 \begin{multline*}
    u \shuffle v  = \{ w \in \Sigma^*  \mid  w = x_1 y_1 x_2 y_2 \cdots x_n y_n 
    \emph{ for some words } \\ x_1, \ldots, x_n, y_1, \ldots, y_n \in \Sigma^*
    \emph{ such that } u = x_1 x_2 \cdots x_n \emph{ and } v = y_1 y_2 \cdots y_n \},
 \end{multline*}
 for $u,v \in \Sigma^{\ast}$ and 
  $L_1 \shuffle L_2  := \bigcup_{x \in L_1, y \in L_2} (x \shuffle y)$ for $L_1, L_2 \subseteq \Sigma^{\ast}$.
 If $L_1, \ldots, L_n \subseteq \Sigma^*$, we set $\bigshuffle_{i=1}^n L_i = L_1 \shuffle \ldots \shuffle L_n$.

Let $\Gamma \subseteq \Sigma$.
The \emph{projection homomorphism} $\pi_{\Gamma} : \Sigma^* \to \Gamma^*$ 
is the homomorphism given by $\pi_{\Gamma}(x) = x$ for $x \in \Gamma$,
$\pi_{\Gamma}(x) = \varepsilon$ otherwise and extended by $\pi_{\Gamma}(\varepsilon) = \varepsilon$
and $\pi_{\Gamma}(wx) = \pi_{\Gamma}(w)\pi_{\Gamma}(x)$ for $w \in \Sigma^*$ and $x \in \Sigma$.
As a shorthand, we set, with respect to a given naming $\Sigma = \{a_1, \ldots, a_k\}$,
$\pi_j = \pi_{\{a_j\}}$. Then $\pi_j(w) = a_j^{|w|_{a_j}}$.
For $L \subseteq \Sigma^*$, we set $\pi_{\Gamma}(L) = \{ \pi_{\Gamma}(u) \mid u \in L \}$.

A quintuple $\mathcal A = (\Sigma, Q, \delta, q_0, F)$ is a \emph{finite deterministic and complete automaton},  
where $\Sigma$ is the \emph{input alphabet},
 $Q$ the \emph{finite set of states}, $q_0 \in Q$
the \emph{start state}, $F \subseteq Q$ the set of \emph{final states} and 
$\delta : Q \times \Sigma \to Q$ is the \emph{totally defined state transition function}.
The transition function $\delta : Q \times \Sigma \to Q$
extends to a transition function on words $\delta^{\ast} : Q \times \Sigma^{\ast} \to Q$
by setting $\delta^{\ast}(q, \varepsilon) := q$ and $\delta^{\ast}(q, wa) := \delta(\delta^{\ast}(q, w), a)$
for $q \in Q$, $a \in \Sigma$ and $w \in \Sigma^{\ast}$. In the remainder, we drop
the distinction between both functions and will also denote this extension by~$\delta$.
Here, we do not consider incomplete automata.
The language \emph{recognized} by an automaton $\mathcal A = (\Sigma, Q, \delta, q_0, F)$ is
$
 L(\mathcal A) = \{ w \in \Sigma^{\ast} \mid \delta(q_0, w) \in F \}.
$
A language $L \subseteq \Sigma^{\ast}$ is called \emph{regular} if $L = L(\mathcal A)$
for some finite automaton~$\mathcal A$. 

A language $L \subseteq \Sigma^*$ is a \emph{group language}, if there exists
a \emph{permutation automaton} $\mathcal A = (\Sigma, Q, \delta, q_0, F)$, i.e., an automaton
such that the map $q \mapsto \delta(q, a)$ is a permutation for each $a \in \Sigma$, recognizing $L$.
A language $L \subseteq \Sigma^*$ is an \emph{aperiodic} language, if there exists an automaton $\mathcal A = (\Sigma, Q, \delta, q_0, F)$ recognizing it
such that, for each $w \in \Sigma^*$, $q \in Q$ and $n \ge 1$, if $\delta(q, w^n) = q$, then $\delta(q, w) = q$.

The \emph{Nerode right-congruence} 
with respect to $L \subseteq \Sigma^*$ is defined, for $u,v \in \Sigma^*$, by $u \equiv_L v$ if and only if 
$
 \forall x \in \Sigma^* : ux \in L \leftrightarrow vx \in L.
$
The equivalence class of $w \in \Sigma^{\ast}$
is denoted by $[w]_{\equiv_L} = \{ x \in \Sigma^{\ast} \mid x \equiv_L w \}$.
A language is regular if and only if the above right-congruence has finite index, and it can
be used to define the \emph{minimal deterministic automaton}
$\mathcal A_L = (\Sigma, Q_L, \delta_L, [\varepsilon]_{\equiv_L}, F_L)$
with 
$Q_L  = \{ [u]_{\equiv_L} \mid u \in \Sigma^{\ast} \}$,
$\delta_L([w]_{\equiv_L}, a)  = [wa]_{\equiv_L}$
and $F_L = \{ [u]_{\equiv_L} \mid u \in L \}$.
Let $L \subseteq \Sigma^*$ be regular
 with minimal automaton $\mathcal A_L = (\Sigma, Q_L, \delta_L, [\varepsilon]_{\equiv_L}, F_L)$.
 The number $|Q_L|$ is called the \emph{state complexity} of $L$
Two words are said to be \emph{distinguishable}, if they denote different right-congruence
classes for a given language.

A language $L \subseteq \Sigma^*$ is \emph{commutative},
if, for $u,v \in \Sigma^*$ such that $|v|_x = |u|_x$ for every $x \in \Sigma$,
we have $u \in L$ if and only if $v \in L$.
\begin{toappendix} 
An automaton $\mathcal A = (\Sigma, Q, \delta, q_0, F)$
is called \emph{commutative}, if, for each $q \in Q$
and $a,b \in \Sigma$, we have $\delta(q, ab) = \delta(q, ba)$.
\end{toappendix}
For commutative regular languages we have the following normal form.

\begin{theorem}[\cite{Hoffmann2021NISextended,DBLP:conf/cai/Hoffmann19}]\label{thm:reg_commutative_form}
 Let $\Sigma = \{a_1, \ldots, a_k\}$ be our alphabet.
 A commutative language $L \subseteq \Sigma^*$
 is regular if and only if it can be written in the form
 $
  L = \bigcup_{i=1}^n U_1^{(i)} \shuffle \ldots \shuffle U_k^{(i)}
 $
 with non-empty unary regular languages $U_j^{(i)} \subseteq \{a_j\}^*$
 for $i \in \{1,\ldots, n\}$ and $j \in \{1,\ldots k\}$
 that can be recognized
 by unary automata with a single final~state.
\end{theorem}

Let $L \subseteq \Sigma^*$ be a commutative regular language. 
For each $j \in \{1, \ldots, k\}$
let $i_j \ge 0$ and $p_j \ge 1$
be the smallest numbers such that $[a_j^{i_j}]_{\equiv_L} = [a_j^{i_j + p_j}]_{\equiv_L}$.
The vectors $(i_1, \ldots, i_k)$
and $(p_1, \ldots, p_k)$ are then called the \emph{index} 
and \emph{period} vectors of $L$.
These notions where introduced in~\cite{GomezA08,Hoffmann2021NISextended,DBLP:conf/cai/Hoffmann19}
and it was shown that they can be used to bound the state complexity of $L$.

\begin{theorem}[\cite{GomezA08,Hoffmann2021NISextended,DBLP:conf/cai/Hoffmann19}]
\label{thm:sc_index_period}
 Let $L \subseteq \Sigma^*$ be a commutative regular language
 with index vector $(i_1, \ldots, i_k)$
 and period vector $(p_1, \ldots, p_k)$.
 Then, for any $j \in \{1,\ldots,k\}$, we have $i_j + p_j \le \stc(L) \le \prod_{r=1}^k (i_r + p_r)$.
\end{theorem}

\begin{example}
\label{ex:index_period}
  Let  $L = (aa)^{\ast} \shuffle (bb)^{\ast} \cup (aaaa)^{\ast} \shuffle b^{\ast}$.
 Then $(i_1, i_2) = (0,0)$, $(p_1, p_2) = (4,2)$,
$\pi_1(L) = (a a)^{\ast}$ and $\pi_2(L) = b^{\ast}$.
\end{example}

The following result from~\cite{Hoffmann2021NISextended,DBLP:conf/cai/Hoffmann19} connects the index and period vector
with the aperiodic and group languages.

\begin{theorem}
\label{thm:char_grp_aperiodic_by_index_period_vectors}
 A commutative regular language is:
 \begin{enumerate}
 \item aperiodic iff its period vector equals $(1,\ldots, 1)$;
 \item a group language iff its index vector equals $(0,\ldots,0)$.
 \end{enumerate}
\end{theorem}




 Let $u, v \in \Sigma^*$.
 Then, $u$ is a \emph{subsequence}\footnote{Also called a \emph{scattered subword}
 in the literature~\cite{DBLP:journals/tcs/GruberHK07,KarandikarNS16}.} of $v$, denoted by $u \preccurlyeq v$,
 if and only if 
 $
   v \in u \shuffle \Sigma^*.
 $
 The thereby given order is called the \emph{subsequence order}.
 Let $L \subseteq \Sigma^*$.
 Then, we define:
\begin{enumerate} 
\item the \emph{upward closure}, by 
 $\mathop{\uparrow\!} L = L \shuffle \Sigma^* = \{ u \in \Sigma^* : \exists v \in L : v \preccurlyeq u \}$;

\item the \emph{downward closure}, by $\mathop{\downarrow\!} L = \{ u \in \Sigma^* : u \shuffle \Sigma^* \cap L \ne \emptyset \} = \{ u \in \Sigma^* : \exists v \in L : u \preccurlyeq v \}$;

\item the \emph{upward interior}, denoted by $\mathop{\downtouparrow\!} L$,
 as the largest upward-closed set in $L$, i.e. the largest
 subset $U \subseteq L$ such that $\mathop{\uparrow\!} U = U$;

\item the \emph{downward interior}, denoted by $\mathop{\uptodownarrow\!} L$,
 as the largest downward-closed set in $L$, i.e., the largest
 subset $U \subseteq L$ such that $\mathop{\downarrow\!} U = U$.
\end{enumerate}
The following equations are valid~\cite{KarandikarNS16}:
\begin{equation}\label{eqn:interior}
 \mathop{\uptodownarrow\!} L = \Sigma^* \setminus \mathop{\uparrow\!} (\Sigma^* \setminus L) \quad 
 \mathop{\downtouparrow\!} L = \Sigma^* \setminus \mathop{\downarrow\!} (\Sigma^* \setminus L).
\end{equation}

A remarkable fact is that for \emph{every} language, the above closure operators
give regular languages. This is based on the fact that the subsequence order
is a well-order, i.e., any upward-closed set is generated by a finite subset of words~\cite{Hai69a,Hig52}.



For general $L \subseteq \Sigma^*$, there is no way to compute a recognizing
automaton for the upward and downward operations.
This is seen by an easy argument, due to~\cite{DBLP:journals/tcs/GruberHK07}.
Let $L \subseteq \Sigma^*$ be a recursively enumerable language.
Then, $L$ is non-empty if and only if $\mathop{\uparrow\!} L$ is non-empty.
However, the former problem is undecidable for recursively enumerable languages,
but decidable for regular languages. Hence, if we can compute a recognizing 
automaton for $\mathop{\uparrow\!} L$, we can decide non-emptiness for $L$, which 
is, in general, not possible.
But for regular and context-free $L$, recognizing automata
for these operations
are computable~\cite{DBLP:journals/tcs/GruberHK07,DBLP:journals/dm/Leeuwen78}.

In~\cite{DBLP:journals/ita/Heam02}, a lower bound for the state complexity of the upward closure
was established by using the language $L = \bigcup_{a \in \Sigma} \{a^N\}$, which
is commutative and finite. As every finite language is aperiodic and 
we are interested in this class, let us highlight this fact
with the next statement.

\begin{proposition}[H\'eam~\cite{DBLP:journals/ita/Heam02}]
\label{prop:aperiodic:upward_lower_bound}
 Set 
 \[ 
 g(n) = \max\{ \stc(\mathop{\uparrow\!} L) \mid \stc(L) \le n \mbox{ and } L\mbox{ is finite and commutative } \}.
 \]
 Then $g(n) \in \Omega\left( \left( \frac{n}{|\Sigma|} \right)^{|\Sigma|} \right)$.
\end{proposition}

The next is from~\cite{Hoffmann2021NISextended,DBLP:conf/cai/Hoffmann19}.

\begin{theorem}
\label{thm::sc_aperiodic}
 Let $U, V \subseteq \Sigma^*$ be aperiodic commutative languages with index vectors $(i_1, \ldots, i_k)$ and  $(j_1, \ldots, j_k)$. 
 Then, $U \shuffle V$ has
 index vector component-wise less than $(i_1 + j_1, \ldots, i_k + j_k)$
 and period vector $(1,\ldots, 1)$. So,
 $
  \stc(U\shuffle V) \le \prod_{l=1}^k (i_l + j_l + 1).
 $ 
 Hence, Theorem~\ref{thm:sc_index_period}
 yields $\stc(U\shuffle V) \le (\stc(U) + \stc(V) - 1)^{|\Sigma|}$.
\end{theorem}

As a corollary of Theorem~\ref{thm::sc_aperiodic} and Theorem~\ref{thm:char_grp_aperiodic_by_index_period_vectors}, we also get, as the star-free and aperiodic languages coincide~\cite{Schutzenberger65a}, an old result by J.F.Perrot~\cite{DBLP:journals/fuin/CastiglioneR12,DBLP:journals/tcs/Perrot78,DBLP:conf/lata/Restivo15}.

\begin{corollary}[J.F.Perrot~\cite{DBLP:journals/tcs/Perrot78}]
 The shuffle of two commutative star-free languages is star-free.
\end{corollary}

\section{The Upward and Downward Closure Operations}

In this section, we establish state complexity bounds for the upward and downward closure and interior operations. The constructions also yield
polynomial time algorithms for computing those closures, if the alphabet is fixed
and not allowed to vary with the input.
\todo{das allgemeine ergebnis für shuffle in preliminaries nennen, und dann den teil wieder auskommentieren?}

\begin{toappendix}
The \emph{Parikh map} is the map $\psi : \Sigma^* \to \mathbb N_0^k$ given by $\psi(w) = (|w|_{a_1}, \ldots, |w|_{a_k})$
for $w \in \Sigma^*$.
If $L$ is commutative, then $L$ can be identified with its Parikh image
$\psi(L)$ and we have
\begin{equation}\label{eqn:L_comm_upward_downward}
    u \in \mathop{\downarrow\!} L \Leftrightarrow \exists v \in L : \psi(u) \le \psi(v) \mbox{ and } 
    u \in \mathop{\uparrow\!} L \Leftrightarrow \exists v \in L : \psi(v) \le \psi(u).
\end{equation}

For $L \subseteq \Sigma^*$, we set $\perm(L) = \psi^{-1}(\psi(L))$, the \emph{commutative closure} of $L$.

We need the following construction
from~\cite{GomezA08}.

Let $\Sigma = \{a_1, \ldots, a_k \}$ be our finite alphabet\footnote{Recall that
we might write $\Sigma = \{a,b\}$ as a shorthand for $\Sigma = \{a_1, a_2\}$
in examples.}. The minimal
commutative automaton\footnote{The minimal commutative automaton is a commutative automaton
in the sense of~\cite{BrzozowskiS73}. However, also the minimal automaton
of a commutative language is commutative in this sense~\cite{DBLP:journals/jalc/Fernau019}.
So, naming the construction to come the minimal commutative automaton
might be a bit misleading, probably minimal canonical commutative automaton being
a better choice. Nevertheless, I will stick to the terminology as introduced in~\cite{GomezA08}.}
for a commutative language
was introduced in~\cite{GomezA08}.

\begin{definition}[minimal commutative automaton]
\label{def::min_com_aut}
 Let $L \subseteq \Sigma^*$. The \emph{minimal commutative  automaton} 
 is $\mathcal C_L = (\Sigma, S_1 \times \ldots \times S_k, \delta, s_0, F)$
 with 
\[
 S_j = \{ [a_j^m]_{\equiv_L} : m \ge 0 \}, \quad
 F  = \{ ([\pi_1(w)]_{\equiv_L}, \ldots, [\pi_k(w)]_{\equiv_L}) : w \in L \}
\]
 and $\delta((s_1, \ldots, s_j, \ldots, s_k), a_j) = (s_1, \ldots, \delta_{j}(s_j, a_j), \ldots, s_k)$
 with one-letter transitions $\delta_{j}([a_j^m]_{\equiv_L}, a_j) = [a_j^{m+1}]_{\equiv_L}$ for $j = 1,\ldots, k$ and $s_0 = ([\varepsilon]_{\equiv_L}, \ldots, [\varepsilon]_{\equiv_L})$. 
\end{definition}
 
\begin{remark}
\label{rem:equal_states_C_L}
 Let $L \subseteq \Sigma^*$ be commutative and $\mathcal C_L = (\Sigma, S_1 \times \ldots \times S_k, \delta, s_0, F)$
 be the minimal commutative automaton.
 Note that, by the definition of the transition function in Definition~\ref{def::min_com_aut},
 we have, for any $u,v \in \Sigma^*$,
 \begin{equation}
    \delta(s_0, u) = \delta(s_0, v) \Leftrightarrow \forall j \in \{1,\ldots,k\} : \pi_j(u) \equiv_L \pi_j(v).
 \end{equation}
\end{remark}

If $L$ is commutative, in~\cite{GomezA08} it was shown that this notion is well-defined
and recognizes $L$, and it was noted that in general the minimal commutative automaton is not equal to the 
minimal deterministic and complete automaton for $L$.

\begin{theorem}[G{\'{o}}mez \& Alvarez~\cite{GomezA08}]
\label{thm::min_com_aut}
 Let $L \subseteq \Sigma^*$ be a commutative language.
 Then, $L = L(\mathcal C_L)$
 and $L$ is regular if and only if $\mathcal C_L$ is finite.
\end{theorem} 

\end{toappendix}

\begin{theoremrep}
\label{thm:sc_closure_interior}
 Let $\Sigma = \{a_1, \ldots, a_k\}$.
 Suppose $L \subseteq \Sigma^*$ is commutative and regular
 with index vector $(i_1, \ldots, i_k)$
 and period vector $(p_1, \ldots, p_k)$.
 Then,
 \[ \max\{\stc(\mathop{\uparrow\!} L), \stc(\mathop{\downarrow\!} L),\stc(\mathop{\downtouparrow\!} L ), \stc(\mathop{\uptodownarrow\!} L) \} \le \prod_{j=1}^k (i_j + p_j).
 \]
\end{theoremrep}
\begin{proofsketch}
 We only give a rough outline of the proof idea
 for the first operation.
 For $L$, as shown in~\cite{GomezA08,Hoffmann2021NISextended,DBLP:conf/cai/Hoffmann19}, we can construct
 an automaton of size $\prod_{j=1}^k (i_j + p_j)$ for $L$
 whose states can be put in correspondence with the modulo and threshold counting
 for the different letters. 
 More formally, it could be shown
 that there exists an automaton
 $\mathcal C = (\Sigma, S_1, \times \ldots \times S_k, \delta, s_0, F)$
 with
 $S_j = \{ 0, \ldots, i_j + p_j - 1 \}$ for $j \in \{1,\ldots, k \}$
 recognizing $L$.

 Construct automaton for $\mathop{\uparrow\!} L$: Set $\mathcal A_{\uparrow} = (\Sigma, S_1 \times \ldots \times S_k, \delta_{\uparrow}, s_0, F_{\uparrow})$
 with
 \begin{align*}
      F_{\uparrow} & = \{ (s_1, \ldots, s_k) \in S_1 \times \ldots \times S_k \mid 
      \\ & \qquad\quad \exists (f_1, \ldots, f_k) \in F \  \forall j \in \{1,\ldots,k\} : f_j \le s_j \}; \\ 
      \delta_{\uparrow}( (s_1, \ldots, s_k), a_j ) & = \left\{ 
  \begin{array}{ll}
   (s_1, \ldots, s_k)              & \mbox{ if } i_j + p_j - 1 = s_j; \\
   \delta((s_1, \ldots, s_k), a_j) & \mbox{ otherwise,} 
  \end{array} 
  \right. 
 \end{align*}
 for $j \in \{1,\ldots,k\}$ and $(s_1, \ldots, s_k) \in S_1 \times\ldots\times S_k$.
 With a similar idea,
 an automaton $\mathcal A_{\downarrow}$
 such that $L(\mathcal A_{\downarrow}) = \mathop{\downarrow\!} L$
 can be constructed.
 See Example~\ref{ex:closure_construction}
 for concrete constructions. \qed
\end{proofsketch}
\begin{proof}
 Let $\mathcal C_L = (\Sigma, S_1, \times \ldots \times S_k, \delta, s_0, F)$
 be the minimal commutative automaton of $L$.
 For notational simplicity, we identify its states
 with tuples of natural numbers. More precisely, we assume
 $S_j = \{ 0, \ldots, i_j + p_j - 1 \}$ for $j \in \{1,\ldots, k \}$, where $n \in S_j$ corresponds to the Nerode right-congruence class $[a_j^n]_{\equiv_L}$. With this identification\footnote{Recall
 that $\mathbb N_0^k$ is equipped with the component-wise order.},
 for $w \in \Sigma^*$ with $\psi(w) \le (i_1 + p_1 - 1, \ldots, i_k + p_k - 1)$,
 we have 
 \begin{equation} \label{eqn:state_identification}
     \delta(s_0, w) = \psi(w).
 \end{equation}
 And for arbitrary $u \in \Sigma^*$,
 \begin{equation} \label{eqn:state_always_smaller_than_parikh_image}
  \delta(s_0, u) \le \psi(u).
 \end{equation}
 Also, if $\delta(s_0, w) = (s_1, \ldots, s_k)$, then, for $j \in \{1,\ldots,k\}$,
  \begin{equation}\label{eqn:pi_j_equals_s_j} 
   ( ( s_j < i_j ) \lor ( |\pi_j(w)| < i_j + p_j) )
   \Rightarrow ( s_j = |\pi_j(w)| \land \delta(s_0, \pi_j(w)) = |\pi_j(w)| ).
  \end{equation}

 \begin{enumerate} 
 \item 
 Set $\mathcal A_{\uparrow} = (\Sigma, S_1 \times \ldots \times S_k, \delta_{\uparrow}, s_0, F_{\uparrow})$
 with, for $j \in \{1,\ldots,k\}$ and $(s_1, \ldots, s_k) \in S_1 \times\ldots\times S_k$,
 \begin{align}
      F_{\uparrow} & = \{ s \in S_1 \times \ldots \times S_k \mid \exists f \in F : f \le s \}; \label{eqn:F_uparrow} \\
      \delta_{\uparrow}( (s_1, \ldots, s_k), a_j ) & = \left\{ 
  \begin{array}{ll}
   (s_1, \ldots, s_k)              & \mbox{ if } i_j + p_j - 1 = s_j; \\
   \delta((s_1, \ldots, s_k), a_j) & \mbox{ otherwise.} 
  \end{array} 
  \right. \label{eqn:delta_uparrow}
 \end{align}
 Then, $\mathcal A_{\uparrow}$ has the following properties: for any $u,v \in \Sigma^*$,
 \begin{equation}
 \label{eqn:A_uparrow}
  \psi(u) \le \psi(v) \Rightarrow \delta_{\uparrow}(s_0, u) \le \delta_{\uparrow}(s_0, v)
  \end{equation}
  and, more generally, for $w \in \Sigma^*$,
  \begin{equation} \label{eqn:A_upparow_transition}
   \delta_{\uparrow}(s_0, w) = (\max\{ |w|_{a_1}, i_1 + p_1 - 1 \}, \ldots, \max\{ |w|_{a_k}, i_k + p_k - 1 \}).
  \end{equation}
  Equation~\eqref{eqn:A_upparow_transition} also implies, for any $w \in \Sigma^*$,
  \begin{equation}\label{eqn:word_greater}
    \delta_{\uparrow}(s_0, w) \le \psi(w).
  \end{equation}

 \begin{claiminproof}
  We have $\operatorname{\uparrow\!} L = L(\mathcal A_{\uparrow})$.
 \end{claiminproof}
 \begin{claimproof}
  Suppose $w \in L(\mathcal A_{\uparrow})$.
  Then, $\delta_{\uparrow}(s_0, w) \in F_{\uparrow}$.
  So, there exists $f \in F$ with $f \le \delta_{\uparrow}(s_0, w)$.
  Let $u \in \Sigma^*$ be minimal with $\delta_{\uparrow}(s_0, u) = f$.
  Then, by Equation~\eqref{eqn:A_upparow_transition} and the minimality of $u$,
  we have, for any $j \in \{1,\ldots,k\}$, that $0 \le |u|_{a_j} \le i_j + p_j - 1$,
  or $\psi(u) \le (i_1 + p_1 - 1, \ldots, i_k + p_k - 1)$.
  So, by Equation~\eqref{eqn:A_upparow_transition},
  we find $f = \delta_{\uparrow}(s_0, u) = \psi(u)$.
  Also, by Equation~\eqref{eqn:state_identification},
  $\delta(s_0, u)  = \psi(u) = f$.
  Hence, $u \in L$.
  Then, combining the previous facts and using Equation~\eqref{eqn:word_greater},
  \[
   \psi(u) = f \le \delta_{\uparrow}(s_0, w) \le \psi(w).
  \]
  So, by Equation~\eqref{eqn:L_comm_upward_downward}, $w \in \mathop{\uparrow\!} L$.

  Conversely, suppose $w \in \mathop{\uparrow\!} L = L \shuffle \Sigma^*$.
  Then, some minimal subsequence $u$ 
  of $w$ drives $\mathcal C_L$ into a final state,
  i.e., $\delta(s_0, u) \in F$. We have $\psi(u) \le \psi(w)$.
  Hence, by Equation~\eqref{eqn:A_uparrow}, 
  \[
   \delta_{\uparrow}(s_0, u) \le \delta_{\uparrow}(s_0, w).
  \] 
  Assume $|u|_{a_j} > i_j + p_j - 1$ for some $j \in \{1,\ldots,k\}$.
  Set $m = |u|_{a_j} - p_j$.
  Then, $\delta(s_0, u) = \delta(s_0, \pi_1(u) \cdots \pi_{j-1}(u) a_j^m \pi_{j+1}(u) \cdots \pi_k(u))$
  and so, by commutativity, by leaving out $p_j$ many times the letter $a_j$
  in $u$ we find a shorter word that ends in the same state.
  However,  this is excluded by minimality of $u$. Hence, $\psi(u) \le (i_1 + p_1 - 1, \ldots, i_k + p_k - 1)$.
  So, by Equation~\eqref{eqn:A_upparow_transition},
  $\psi(u) = \delta_{\uparrow}(s_0, u)$
  and, by Equation~\eqref{eqn:state_identification},
  $\delta(s_0, u) = \psi(u)$.
  Therefore,
  \[
   \delta_{\uparrow}(s_0, u) = \psi(u) = \delta(s_0, u) \in F.
  \]
  As $\delta_{\uparrow}(s_0, u) \le \delta_{\uparrow}(s_0, w)$
  and $\delta_{\uparrow}(s_0, u) \in F$,
  we have  $\delta_{\uparrow}(s_0, w) \in F_{\uparrow}$.
  Hence, $w \in L(\mathcal A_{\uparrow})$.
 \end{claimproof}
 
 \medskip
 
 \item 
 Set $\mathcal A_{\downarrow} = (\Sigma, S_1\times \ldots\times S_k, \delta, s_0, F_{\downarrow})$
 with\footnote{Observe that we retain the transition function of $\mathcal C_L$. In $\mathcal A_{\downarrow}$, only the final
 state set is altered.}
 \[
  F_{\downarrow} = \{ s \in S_1 \times\ldots\times S_k \mid \exists e \in E' : s \le e \},
 \]
 where 
 \begin{align*}
     E' & = \{ (s_1, \ldots, s_k) \in S_1 \times\ldots\times S_k \mid \\ 
        & \qquad \exists (f_1, \ldots, f_k) \in F \ \forall j \in \{1, \ldots, k\} : \\ 
        & \qquad ( ( i_j \le s_j < i_j + p_j ) \land ( i_j \le f_j < i_j + p_j ) ) \lor ( f_j = s_j )\}.
 \end{align*}

 Note that, $(s_1, \ldots, s_k) \in E'$ if and only if there exists $(f_1, \ldots, f_k) \in F$
 such that, for any $j \in \{1,\ldots,k\}$, either
 \[
  ( i_j \le s_j < i_j + p_j ) \land ( i_j \le f_j < i_j + p_j ) 
 \]
 or
 \begin{equation}\label{eqn:max_i}
  ( \max\{ f_j, s_j \} < i_j ) \land ( f_j = s_j ).
 \end{equation}

 \begin{claiminproof}
  We have $\mathop{\downarrow\!} L = L(\mathcal A_{\downarrow})$.
 \end{claiminproof}
 \begin{claimproof}
  Suppose  $w \in L(\mathcal A_{\downarrow})$. 
  Set $(s_1, \ldots, s_k) = \delta(s_0, w)$. We
  have two cases:
  
  \begin{enumerate}
  \item $(s_1, \ldots, s_k) \in E'$.
  
   \medskip 
   
    We first show the auxiliary statement (i), and use this
    to construct a word $u \in L$ with $w \preccurlyeq u$.
     
     \begin{enumerate}
     \item[(i)] We claim that there exists $(f_1, \ldots, f_k) \in F$  such that, for any $j \in \{1,\ldots,k\}$, we can find $m_j \ge 0$
    with $m_j \ge |\pi_j(w)| \ge s_j$ and
    \begin{equation}\label{eqn:m_j}
    \delta(s_0, a_j^{m_j}) = ([\varepsilon]_{\equiv_L}, \ldots, [\varepsilon]_{\equiv_L}, f_j, [\varepsilon]_{\equiv_L}, \ldots, [\varepsilon]_{\equiv_L}).
    \end{equation}
    
     \emph{Proof of (i):} Intuitively, if $(s_1, \ldots, s_k) \in E'$, for some state 
     from $F$, each entry $s_j$, $j \in \{1,\ldots, k\}$,
     either equals the $j$-th entry in the state  from $F$ or shares a cycle with the $j$-th
     entry in the state  from $F$.
     More formally, by the definition of $E'$,
     there exists $f = (f_1, \ldots, f_k) \in F$ such that, for any $j \in \{1,\ldots,k\}$,
     \[
         ( ( i_j \le s_j < i_j + p_j ) \land ( i_j \le f_j < i_j + p_j ) ) \mbox{ or } ( s_j = f_j ).
     \]
     Set $r = |w|$ and define a vector $(m_1, \ldots, m_k) \in \mathbb N^k$ according to the following rules,
     based on the definition of~$E'$,
     \[
      m_j = \left\{ 
      \begin{array}{ll}
       f_j       & \mbox{if } f_j = s_j; \\
       f_j + r\cdot p_j & \mbox{if } ( i_j \le s_j < i_j + p_j ) \land ( i_j \le f_j < i_j + p_j ).
      \end{array}\right.
     \]
     Then, Equation~\eqref{eqn:m_j} holds true. By Equation~\eqref{eqn:state_always_smaller_than_parikh_image},
     we have $s_j \le |\pi_j(w)|$. We only need to show $m_j \ge |\pi_j(w)|$.
     Let $j \in \{1,\ldots,k\}$.
     If $( i_j \le s_j < i_j + p_j ) \land ( i_j \le f_j < i_j + p_j )$, as $p_j > 0$
     and by choice of $r$
     we have $m_j \ge |\pi_j(w)|$.
     Otherwise, we must have $s_j = f_j$ and, by Equation~\eqref{eqn:max_i},
     $s_j < i_j$. So, by Equation~\eqref{eqn:pi_j_equals_s_j}  and the definition of $m_j$,
     $m_j = f_j = s_j = |\pi_j(w)|$. \emph{[End, Proof of (i)]}
     \end{enumerate}
     
    \medskip
     
    Set $u = a_1^{m_1} \cdot\ldots\cdot a_k^{m_k}$.
    By Definition~\ref{def::min_com_aut} and Equation~\eqref{eqn:m_j}, 
    $\delta(s_0, u) = (f_1, \ldots, f_k)$.
    Hence, $u \in L$, and, as $m_j \ge |\pi_j(w)|$, $j \in \{1,\ldots,k\}$,
    we find $w \shuffle \Sigma^* \cap \perm(u) \ne \emptyset$,
    and, as $\perm(u) \subseteq L$, $w \shuffle \Sigma^* \cap L \ne \emptyset$.
    Hence, $w \in \mathop{\downarrow\!} L$.
    
   \medskip
 
  \item $(s_1, \ldots, s_k) \in F_{\downarrow}\setminus E'$.
  
   \medskip 
   
    Then, we find $e = (e_1, \ldots, e_k) \in E'$
    such that $\delta(s_0, w) \le e$.
    Write $e = (e_1,\ldots, e_k)$ and $\delta(s_0, w) = (s_1,\ldots, s_k)$.
    Set \[ u = wa_1^{e_1 - s_1} \cdot\ldots\cdot a_k^{e_k - s_k}. \]
    Then, $\delta(s_0, u) \in E'$.
    By the first case, as this handled the case of
    words that end in a state in $E'$, 
    we can deduce $u \in \mathop{\downarrow\!} L$.
    But then, as $w \preccurlyeq u$, we also have $w \in \mathop{\downarrow\!} L$.
  \end{enumerate}
  
  \medskip 
  
 Conversely, suppose $w \in \mathop{\downarrow\!} L$.
 Then, there exists $u \in L$ such that $u \in w \shuffle \Sigma^*$, which
 implies, for any $j \in \{1,\ldots,k\}$, $|\pi_j(w)| \le |\pi_j(u)|$.
 Set
 \[
  (f_1, \ldots, f_k) = \delta(s_0, u) \mbox{ and } (s_1, \ldots, s_k) = \delta(s_0, w).
 \]
 We define a state $t = (t_1, \ldots, t_k) \in S_1 \times \ldots \times S_k$ 
 as follows. Let $j \in \{1,\ldots, k\}$,
 then define $t_j$ according to the following cases:
 \begin{enumerate}
 \item $f_j < i_j$.  Then, as $|\pi_j(w)| \le |\pi_j(u)|$ and $f_j < i_j$, which also
 implies $s_j < i_j$, we have, by the definition of the transition
 function of $\mathcal C_L$, then $s_j \le f_j < i_j$.
  Set $t_j = f_j$.
 
 \item $i_j \le f_j < i_j + p_j$. Set
  \[
   t_j = \left\{ \begin{array}{ll}
    f_j & \mbox{if } s_j < i_j; \\
    s_j & \mbox{if } i_j \le s_j < i_j + p_j.
   \end{array}\right.
  \]
 \end{enumerate}
 Then, for any $j \in \{1,\ldots,k\}$,
 \[
  ( ( i_j \le t_j < i_j + p_j ) \land ( i_j \le f_j < i_j + p_j ) ) \mbox{ or } ( t_j = f_j ).
 \]
 Therefore, $t \in E'$.
 By definition of $t$, we have $s \le t$. So, we find $s \in F_{\downarrow}$.
 Hence $w \in L(\mathcal A_{\downarrow})$.
 \end{claimproof}
 
 \item Due to Equation~\eqref{eqn:interior},
  by switching the final state set of an automaton for $\mathop{\downarrow\!} (\Sigma^* \setminus L)$,
  we get a recognizing automaton for $\mathop{\downtouparrow\!} L$. 
  So, this is implied by the state complexity bound for $\mathop{\downarrow\!} L$.
 
 \item By Equation~\eqref{eqn:interior},
 similarly as for  $\mathop{\downtouparrow\!} L$, the state complexity bound
 is implied.
 \end{enumerate}
 So, all the state complexity bounds are established.~\qed
\end{proof}

\begin{example}
\label{ex:closure_construction}
 Let $L = bb(bb)^* \cup ( b \shuffle a(aa)^* )$.
 Then, $\mathop{\uparrow\!} L = bbb^* \cup ( bb^* \shuffle aa^* )$,
 $\mathop{\downarrow\!} L = b^* \cup ( aa^* \shuffle \{\varepsilon,b\} )$,
 $\mathop{\downtouparrow\!} L   = \emptyset$ and
 $\mathop{\uptodownarrow\!} L = \emptyset$.
 The constructions of the automata $\mathcal A_{\uparrow}$
 and $\mathcal A_{\downarrow}$ 
 with $L(\mathcal A_{\uparrow}) = \mathop{\uparrow\!} L$
 and $L(\mathcal A_{\downarrow}) = \mathop{\downarrow\!} L$
 from the proof sketch of Theorem~\ref{thm:sc_closure_interior}
 are illustrated, for the example language $L$, in Figure~\ref{fig:closure_construction}.
 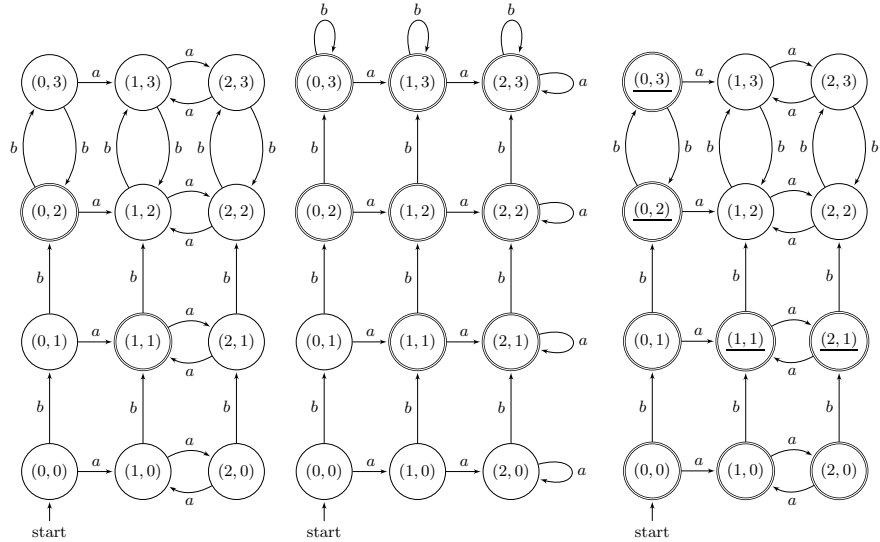
\begin{figure}[h]
     \centering
    \scalebox{.69}{    
 \begin{tikzpicture}[>=latex',shorten >=1pt,node distance=2.5cm and 1.8cm,on grid,auto]
  \node[state, initial below]            (1)                {$(0,0)$};
  \node[state]                     (2) [right =of 1]   {$(1,0)$};
  \node[state]                     (3) [right =of 2]   {$(2,0)$};
  \node[state]                     (4) [above =of 1]   {$(0,1)$};
  \node[state,accepting]           (5) [right =of 4]   {$(1,1)$};
  \node[state]                     (6) [right =of 5]   {$(2,1)$}; 
  \node[state,accepting]           (7) [above =of 4]   {$(0,2)$};
  \node[state]                     (8) [right =of 7]   {$(1,2)$};
  \node[state]                     (9) [right =of 8]   {$(2,2)$}; 
  \node[state]                     (10) [above =of 7]  {$(0,3)$};
  \node[state]                     (11) [right =of 10] {$(1,3)$};
  \node[state]                     (12) [right =of 11] {$(2,3)$}; 

  \path[->] (1) edge node {$a$} (2)
            (2) edge [bend left] node {$a$} (3)
            (3) edge [bend left] node {$a$} (2);

  \path[->] (4) edge node {$a$} (5)
            (5) edge [bend left] node {$a$} (6)
            (6) edge [bend left] node {$a$} (5);
            
  \path[->] (7) edge node {$a$} (8)
            (8) edge [bend left] node {$a$} (9)
            (9) edge [bend left] node {$a$} (8);
    
  \path[->] (10) edge node {$a$} (11)
            (11) edge [bend left] node {$a$} (12)
            (12) edge [bend left] node {$a$} (11);

  \path[->] (1) edge node {$b$} (4)
            (4) edge node {$b$} (7)
            (7) edge [bend left] node {$b$} (10)
           (10) edge [bend left] node {$b$} (7);
 
  \path[->] (2) edge node {$b$} (5)
            (5) edge node {$b$} (8)
            (8) edge [bend left] node {$b$} (11)
           (11) edge [bend left] node {$b$} (8);
           
  \path[->] (3) edge node {$b$} (6)
            (6) edge node {$b$} (9)
            (9) edge [bend left] node {$b$} (12)
           (12) edge [bend left] node {$b$} (9);
  
 \end{tikzpicture}}
 \scalebox{.69}{   
 \begin{tikzpicture}[>=latex',shorten >=1pt,node distance=2.5cm and 1.8cm,on grid,auto]
  \node[state, initial below]            (1)                {$(0,0)$};
  \node[state]                     (2) [right =of 1]   {$(1,0)$};
  \node[state]                     (3) [right =of 2]   {$(2,0)$};
  \node[state]                     (4) [above =of 1]   {$(0,1)$};
  \node[state,accepting]                     (5) [right =of 4]   {$(1,1)$};
  \node[state,accepting]                     (6) [right =of 5]   {$(2,1)$}; 
  \node[state,accepting]                     (7) [above =of 4]   {$(0,2)$};
  \node[state,accepting]                     (8) [right =of 7]   {$(1,2)$};
  \node[state,accepting]                     (9) [right =of 8]   {$(2,2)$}; 
  \node[state,accepting]                     (10) [above =of 7]  {$(0,3)$};
  \node[state,accepting]                     (11) [right =of 10] {$(1,3)$};
  \node[state,accepting]                     (12) [right =of 11] {$(2,3)$}; 
  
  \path[->] (1) edge node {$a$} (2)
            (2) edge node {$a$} (3)
            (3) edge [loop right] node {$a$} (3);

  \path[->] (4) edge node {$a$} (5)
            (5) edge node {$a$} (6)
            (6) edge [loop right] node {$a$} (5);
            
  \path[->] (7) edge node {$a$} (8)
            (8) edge  node {$a$} (9)
            (9) edge [loop right] node {$a$} (8);
    
  \path[->] (10) edge node {$a$} (11)
            (11) edge  node {$a$} (12)
            (12) edge [loop right] node {$a$} (11);

  \path[->] (1) edge node {$b$} (4)
            (4) edge node {$b$} (7)
            (7) edge  node {$b$} (10)
           (10) edge [loop above] node {$b$} (7);
 
  \path[->] (2) edge node {$b$} (5)
            (5) edge node {$b$} (8)
            (8) edge  node {$b$} (11)
           (11) edge [loop above] node {$b$} (8);
           
  \path[->] (3) edge node {$b$} (6)
            (6) edge node {$b$} (9)
            (9) edge  node {$b$} (12)
           (12) edge [loop above] node {$b$} (9);
 \end{tikzpicture}}
 \scalebox{.69}{   
 \begin{tikzpicture}[>=latex',shorten >=1pt,node distance=2.5cm and 1.8cm,on grid,auto]
  \node[state, initial below,accepting]            (1)                {$(0,0)$};
  \node[state,accepting]                     (2) [right =of 1]   {$(1,0)$};
  \node[state,accepting]                     (3) [right =of 2]   {$(2,0)$};
  \node[state,accepting]                     (4) [above =of 1]   {$(0,1)$};
  \node[state,accepting]           (5) [right =of 4]   {\underline{$(1,1)$}};
  \node[state,accepting]           (6) [right =of 5]   {\underline{$(2,1)$}}; 
  \node[state,accepting]           (7) [above =of 4]   {\underline{$(0,2)$}};
  \node[state]                     (8) [right =of 7]   {$(1,2)$};
  \node[state]                     (9) [right =of 8]   {{$(2,2)$}}; 
  \node[state,accepting]                     (10) [above =of 7]  {\underline{$(0,3)$}};
  \node[state]                     (11) [right =of 10] {$(1,3)$};
  \node[state]                     (12) [right =of 11] {$(2,3)$}; 
  
   \path[->] (1) edge node {$a$} (2)
            (2) edge [bend left] node {$a$} (3)
            (3) edge [bend left] node {$a$} (2);

  \path[->] (4) edge node {$a$} (5)
            (5) edge [bend left] node {$a$} (6)
            (6) edge [bend left] node {$a$} (5);
            
  \path[->] (7) edge node {$a$} (8)
            (8) edge [bend left] node {$a$} (9)
            (9) edge [bend left] node {$a$} (8);
    
  \path[->] (10) edge node {$a$} (11)
            (11) edge [bend left] node {$a$} (12)
            (12) edge [bend left] node {$a$} (11);

  \path[->] (1) edge node {$b$} (4)
            (4) edge node {$b$} (7)
            (7) edge [bend left] node {$b$} (10)
           (10) edge [bend left] node {$b$} (7);
 
  \path[->] (2) edge node {$b$} (5)
            (5) edge node {$b$} (8)
            (8) edge [bend left] node {$b$} (11)
           (11) edge [bend left] node {$b$} (8);
           
  \path[->] (3) edge node {$b$} (6)
            (6) edge node {$b$} (9)
            (9) edge [bend left] node {$b$} (12)
           (12) edge [bend left] node {$b$} (9);
 \end{tikzpicture}}
  \caption{Construction of automata
   for the upward and downward closure
   for the language $L = bb(bb)^* \cup (b \shuffle a(aa)^*)$
   by starting from an automaton in a ``rectangular'' normal form for $L$. The left-most automaton recognizes $L$,
   the automaton in the middle recognized
   $\mathop{\uparrow\!} L$
   and the right-most automaton recognized $\mathop{\downarrow\!} L$.
   See Example~\ref{ex:closure_construction} for details.}
   \label{fig:closure_construction}
\end{figure}

\end{example}

The constructions done in the proof of Theorem~\ref{thm:sc_closure_interior}
can actually be performed in polynomial time.

\begin{corollary}
\label{cor:poly_alg_closures}
 Fix the alphabet $\Sigma$.
 Let $L \subseteq \Sigma^*$ be commutative and regular, given 
 by a finite recognizing automaton with $n$ states.
 Then, recognizing automata for $\mathop{\uparrow\!} L$, $\mathop{\downarrow\!} L$,
 $\mathop{\uptodownarrow\!} L$ and $\mathop{\downtouparrow\!} L$
 are computable in polynomial time in $n$.
\end{corollary}

With Theorem~\ref{thm:sc_index_period}, we can derive the next bound from Theorem~\ref{thm:sc_closure_interior}.

\begin{theorem} 
\label{thm:sc_upper_bound_closure_interior}
 Let $L \subseteq \Sigma^*$ be commutative and 
 recognizable by an automaton with $n$ states.
 Then, the upward and downward closures and interiors
 of $L$ are recognizable by automata of size $n^{|\Sigma|}$.
\end{theorem}


\todo{Paragraph auskommentiert}

\section{The Case of Commutative Group Languages}

Before we investigate the state complexity of shuffle, union, intersection and the closure and interior operations
for commutative group languages, we give 
a normal form theorem for commutative groups languages similar to Theorem~\ref{thm:reg_commutative_form}.

\begin{theoremrep}
 Let $\Sigma = \{a_1, \ldots, a_k\}$
 and $L \subseteq \Sigma^*$.
 Then, the following conditions are equivalent:
 \begin{enumerate}
 \item $L$ is a commutative group language; \todo{auch noch klammern um items?}
 \item $L$ is a finite union of languages of the form 
  $
   U_1 \shuffle \ldots \shuffle U_k
  $
  with $U_j \subseteq \{ a_j \}^*$ a group language
  recognizable by an automaton with a single final state;
  \item $L$ is a finite union of languages of the form 
  $ 
   U_1 \shuffle \ldots \shuffle U_k
  $
  with $U_j \subseteq \{ a_j \}^*$ being group languages.
 \end{enumerate}
\end{theoremrep}
\begin{proof}
 Proof that (1) implies (2): Let $L \subseteq \Sigma^*$ be a commutative group language
 with recognizing commutative permutation automaton $\mathcal A = (\Sigma, Q, \delta, q_0, F)$.
 Then, for $q_f \in F$,
 \[
  L((\Sigma, Q, \delta, q_0, \{q_f\})) = \bigshuffle_{j=1}^k \pi_j(L((\Sigma, Q, \delta, q_0, \{q_f\}))).
 \]
 This equation follows easily, the inclusion of the set on the left hand side in the set on the right hand side
 is true for any language. The other inclusion is implied for if we have a word
 in this language, then some permutation of it is in $L((\Sigma, Q, \delta, q_0, \{q_f\}))$,
 but as the automaton is commutative the original automaton also ends up in the state $q_f$.
 By~\cite[Corollary 9]{DBLP:conf/cai/Hoffmann19}, the languages
 $\pi_j(L((\Sigma, Q, \delta, q_0, \{q_f\})))$ are unary group languages.
 These unary languages can be written as unions of languages recognizable by automata
 with a single final state. Then, as shuffle distributes over union,
 the claim follows.

 Proof that (2) implies (3): This is clear.

 Proof that (3) implies (1): By~\cite[Lemma 13]{DBLP:conf/cai/Hoffmann19}, the languages
 $U_1 \shuffle \ldots \shuffle U_k$ are group languages
 and, as the $U_j$ are unary, these languages are commutative
 So, as the group languages and the commutative languages are closed under union~\cite{DBLP:reference/hfl/Pin97},
 we can deduce that $L$ is a commutative group language.\qed 
\end{proof}

\subsection{The Shuffle Operation}

Here, we give a sharp bound for the state complexity
of two commutative group languages.
However, in this case,  we do not express the bound in terms of the size
of recognizing input automata, but in terms of the index and period vectors
of the input languages. The result generalizes a corresponding result from~\cite{PighizziniS02}
for unary group languages to commutative group languages.


\begin{toappendix}
For lower bound results, we also need the next results.

\begin{lemma}
\label{lem:lower_bound_sc}
 Let $\Sigma = \{a_1, \ldots, a_k\}$ and $(n_1, \ldots, n_k) \in \mathbb N_0^k$.
 Suppose for a commutative language $L \subseteq \Sigma^*$ we have
 \begin{enumerate}
 \item $\{ w \in \Sigma^* \mid \forall j \in \{1,\ldots, k\} : |w|_{a_j} \ge n_j \} \subseteq L$,
 \item $\{ w \in \Sigma^* \mid \exists j \in \{1,\ldots, k\} : |w|_{a_j} = \min\{n_j - 1,0\}  \} \cap L = \emptyset$.
 \end{enumerate}
 Then 
 $
  \stc(L) = \prod_{j=1}^k (n_j + 1)
 $
 with index vector $(n_1, \ldots, n_k)$
 and period vector $(1,\ldots, 1)$.
\end{lemma}
\begin{proof}
 Set $\mathcal C = (\Sigma, [n_1+1] \times \ldots \times [n_k+1], \delta, (0,\ldots, 0), F)$
 with
 \begin{multline*}
  \delta((i_1, \ldots, i_{j-1}, i_j, i_{j+1}, \ldots, i_k), a_j) \\
   = (i_1, \ldots, i_{j-1}, (i_j + 1) \bmod (n_j+1), i_{j+1},\ldots, i_k).
 \end{multline*}
 and $F = \{ (\min\{ |w|_{a_1}, n_1 \}, \ldots, \min\{ |w|_{a_k}, n_k \}) \mid w \in L \}$.
 Then $L(\mathcal C) = L$.
 Let $(i_1, \ldots, i_k), (l_1, \ldots, l_k) \in [n_1+1] \times \ldots \times [n_k+1]$
 be distinct. Suppose, without loss of generality, $i_j > l_j$
 for some $j \in \{1,\ldots, k\}$.
 This implies $0 \le l_j < n_j$.
 Then
 $$
  \delta((l_1, \ldots, l_k), a_1^{n_1} \cdots a_{j-1}^{n_{j-1}} a_j^{n_j - 1 - l_j} a_{j+1}^{n_{j+1}} \cdots a_k^{n_k}) \notin F
 $$
 but
 $$ 
  \delta((i_1, \ldots, i_k), a_1^{n_1} \cdots a_{j-1}^{n_{j-1}} a_j^{n_j - 1 - l_j} a_{j+1}^{n_{j+1}} \cdots a_k^{n_k}) \in F
 $$
 Hence, all states are distinguishable and $\mathcal C$
 is isomorphic to the minimal automaton of $L$, which proves the claim. \qed
\end{proof}

\begin{lemma}(Refining a result from \cite{PighizziniS02})
\label{lem:pighizziniS02_unary_group_lang_sc}
 Let $\mathcal A = (\{a\}, Q, \delta, s_0, F)$
 and $\mathcal B = (\{a\}, P, \mu, t_0, E)$ be two unary automata
 with index zero and periods $p$ and $q$ respectively.
 Write $F = \{ f_1, \ldots, f_n\}$, $E = \{e_1, \ldots, e_m\}$.
 Then $L(\mathcal A) L(\mathcal B)$
 could be accepted by an automaton $\mathcal C = (\{a\}, R, \eta, r_0, T)$
 with index $\lcm(p,q) - 1$, period $\gcd(p,q)$ and 
 $T = \bigcup_{l=1}^n \bigcup_{h=1}^m T_{l,h}$ such
 that
 $$
  \{ w \mid \eta(r_0, w) \in T_{l, h} \} = \{ u \mid \delta(s_0, u) = f_l \} \cdot \{ v \mid \mu(t_0, v) = e_h\}
 $$
 for $l \in \{1,\ldots, n\}$ and $h \in \{1,\ldots, m\}$.
 This result is optimal in the sense that
 there exists automata $\mathcal A$ and $\mathcal B$
 as above such that $\mathcal C$ is isomorphic
 to the minimal automaton of $L(\mathcal A)L(\mathcal B)$.
\end{lemma}
\begin{proof} The existence and optimality was stated in~\cite{PighizziniS02} as Theorem 8,
we only show the additional part about the final states.
With the notation from the statement,
set $\mathcal A_l = (\{a\}, Q, \delta, s_0, \{f_l\})$
and $\mathcal B_h = (\{a\}, Q, \mu, t_0, \{e_h\})$ 
for $l \in \{1,\ldots, n\}$ and $h \in \{1,\ldots,m\}$.
Then by~\cite{PighizziniS02} we have an automaton
with index $\lcm(p,q) - 1$
and period $\gcd(p,q)$ accepting $L(\mathcal A_l) \cdot L(\mathcal B_h)$
for $l \in \{1,\ldots, n\}$ and $h \in \{1,\ldots,m\}$.
As the index and period determines the form
of the automaton uniquely, we can suppose the only differing parts
of those automata for each $l \in \{1,\ldots, n\}$
and $h \in \{1,\ldots, m\}$ are the final states.
Hence, we can write $\mathcal C_{l,h} = (\{a\}, R, \eta, r_0, T_{l, h})$
where $R$, $\eta$ and the start state $r_0$ are independent of $l$ and $h$, and
$$
 L((\{a\}, R, \eta, r_0, T_{l, h})) = L(\mathcal A_l) L(\mathcal B_h).
$$
Set $T = \bigcup_{l = 1}^n \bigcup_{h = 1}^m T_{l,h}$, then
\begin{align*} 
 L((\{a\}, R, \eta, r_0, T))
  & =  \bigcup_{l = 1}^n \bigcup_{h = 1}^m L(\mathcal A_l) L(\mathcal B_h) \\
  & = \left( \bigcup_{l = 1}^n L(\mathcal A_l) \right) \cdot \left( \bigcup_{h = 1}^m L(\mathcal B_h) \right) \\
  & = L(\mathcal A) L(\mathcal B).
\end{align*}
This shows our claim. $\qed$
\end{proof}

\begin{lemma}(Combining results from \cite{Hoffmann2021NISextended,DBLP:conf/cai/Hoffmann19})
\label{lem:unary_aut_compatible_with_union}
 Let $\Sigma = \{a_1,\ldots, a_k\}$ and $L \subseteq \Sigma^*$ be a commutative regular language
 with index vector $(i_1, \ldots, i_k)$
 and period vector $(p_1, \ldots, p_k)$.
 Then we can write
 $$
  L = \bigcup_{l=1}^n U_1^{(l)} \shuffle \ldots \shuffle U_k^{(l)}
 $$
 with unary languages $U_j^{(l)} \subseteq \{a_j\}^*$  for $j \in \{1,\ldots, k\}$.
 Furthermore, we can find unary automata $\mathcal A_j = (\{a_j\}, Q_j, \delta_j, s_j, F_j)$
 with indices $i_j$ and periods $p_j$ and
 with $F_j = \{ f_j^{(1)}, \ldots, f_j^{(n)} \}$
 such that for $l \in \{1,\ldots, n\}$
 $$
  U_j^{(l)} = \{ u \in \{a_j\}^* : \delta_j(u) = f_j^{(l)} \}
 $$
 and $|Q_j| = i_j + p_j$. Hence $L(\mathcal A_j) = \{ a_j^{|u|_{a_j}} \mid u \in L \} = \bigcup_{l=1}^n U_j^{l}$.
\end{lemma}
\begin{proof} 
We refer to~\cite{DBLP:conf/cai/Hoffmann19,Hoffmann2021NISextended,GomezA08}
for the definition of the minimal commutative automaton.
We also use the same notation for the languages $U_j^{(l)}$
as used in~\cite{Hoffmann2021NISextended,DBLP:conf/cai/Hoffmann19}.
The first claim is stated as Corollary~2 in~\cite{DBLP:conf/cai/Hoffmann19,Hoffmann2021NISextended}.
Lemma 5 of~\cite{Hoffmann2021NISextended,DBLP:conf/cai/Hoffmann19}
states that we can derive from the minimal commutative automaton
a unary automaton with $i_j + p_j$ states such that
$U_j^{(l)}$ for $l \in \{1,\ldots, n\}$ is precisely the set of words
that lead this automaton into a single final state.
For each $U_j^{(l)}$ the automaton with $i_j + p_j$ states
is the same automaton. Hence, collecting in $F$ all the final
states corresponding to the $U_j^{(l)}$ gives the claim.
Note that we could have $U_j^{(l)} = U_j^{(l')}$
for distinct $l, l' \in \{1,\ldots, n\}$.
Also note that by definition, and as the minimal commutative automaton
is deterministic, if $U_j^{(l)} \ne U_j^{(l')}$, then 
$U_j^{(l)} \cap U_j^{(l')} = \emptyset$. \qed
\end{proof}
\end{toappendix}

 \begin{theoremrep}
   \label{thm:sc_comm_grp_lang}
     Let $\Sigma = \{a_1,\ldots, a_k\}$. 
     For commutative group languages $U, V \subseteq \Sigma^*$
     with period vectors $(p_1, \ldots, p_k)$
     and $(q_1, \ldots, q_k)$ their shuffle $U\shuffle V$
     has index vector $(i_1, \ldots, i_k)$ with $i_j = \lcm(p_j, q_j) - 1$ 
     for $j \in \{1,\ldots, k\}$
     and period vector 
     $(\gcd(p_1, q_1), \ldots, \gcd(p_k, q_k))$. Hence, by Theorem~\ref{thm:sc_index_period},
     $$
      \stc(U\shuffle V) \le \prod_{j=1}^k ( \gcd(p_j, q_j) + \lcm(p_j, q_j) - 1).
     $$
     Furthermore,
     there exist commutative group languages
     such that a minimal automaton recognizing their shuffle reaches the bound.
    \end{theoremrep}
    \begin{proof}
 Write
 \begin{align*}
     U & = \bigcup_{l=1}^n U_1^{(l)} \shuffle \ldots \shuffle U_k^{(l)} \\
     V & = \bigcup_{h=1}^m V_1^{(h)} \shuffle \ldots \shuffle V_k^{(h)}
 \end{align*}
 with, for $j \in \{1,\ldots, k\}$, unary automata $\mathcal A_j = (\{a_j\}, Q_j, \delta_j, s_j, F_j)$
 and $\mathcal B_j = (\{a_j\}, P_j, \mu_j, t_j, E_j)$, $F_j = \{ f_j^{(l)} \mid l \in \{1,\ldots, n\} \}$
 and $E_j = \{ e_j^{(h)} \mid h \in \{1,\ldots, m\} \}$, with indices zero
 and periods $p_j$ and $q_j$ respectively,
 according to Lemma~\ref{lem:unary_aut_compatible_with_union}.
 Let $\mathcal C_j = (\{a_j\}, R_j, \eta_j, r_j, T_j)$
 with $T_j = \bigcup_{l=1}^n \bigcup_{h=1}^m T_j^{(l,h)}$
 be automata according to Lemma~\ref{lem:pighizziniS02_unary_group_lang_sc}
 such that
 \begin{equation}\label{eqn:concat_parts}
  \{ u \mid \eta_j(r_j, u) \in  T_j^{(l,h)} \} = U_j^{(l)} \cdot V_j^{(h)}
 \end{equation}
 for $l \in \{1,\ldots, n\}$ and $h \in \{1,\ldots, m\}$. 
 Hence, the indices and periods of these automata, as stated in Lemma~\ref{lem:pighizziniS02_unary_group_lang_sc},
 are precisely $\lcm(p_j, q_j) - 1$ and $\gcd(p_j, q_j)$.
 Define
 $$
  \mathcal C = (\Sigma, R_1 \times \ldots \times R_k, \eta, (r_1, \ldots, r_k), T)
 $$
 with
 $$
  T = \bigcup_{l=1}^n \bigcup_{h=1}^m T_1^{(l,h)} \times \ldots \times T_k^{(l,h)}
 $$
 and transition function $\delta((r_1, \ldots, r_j, \ldots, r_k), a_j) = (r_1, \ldots, \delta_j(r_j, a_j), \ldots, r_k)$
 for $(r_1, \ldots, r_j, \ldots, r_k) \in R_1 \times \ldots \times R_k$.
 By construction of $\mathcal C$ and Equation~\eqref{eqn:concat_parts}, we have
 $$
  \delta(u, (r_1, \ldots, r_k)) \in T_1^{(l,h)} \times \ldots \times T_k^{(l,h)}
   \Leftrightarrow u \in ( U_1^{(l)} V_1^{(h)} ) \shuffle \ldots \shuffle ( U_k^{(l)} V_k^{(h)} )
 $$
 for $l \in \{1,\ldots, n\}$ and $h \in \{1,\ldots, m\}$. Hence
 \begin{align*}
  L(\mathcal C) 
   & = \bigcup_{l=1}^n \bigcup_{h=1}^m ( U_1^{(l)} V_1^{(h)} ) \shuffle \ldots \shuffle ( U_k^{(l)} V_k^{(h)} ) \\
   & = \bigcup_{l=1}^n \bigcup_{h=1}^m  U_1^{(l)} \shuffle \ldots \shuffle U_k^{(l)} \shuffle V_1^{(l)} \shuffle \ldots \shuffle V_k^{(l)}\\
   & = \left( \bigcup_{l=1}^n  U_1^{(l)} \shuffle \ldots \shuffle U_k^{(l)} \right) \shuffle \left( \bigcup_{h=1}^m  V_1^{(l)} \shuffle \ldots \shuffle V_k^{(l)} \right) \\
   & = U \shuffle V
 \end{align*}
 as the shuffle operation is commutative and distributive over union and
 for unary languages shuffle and concatenation are the same operations.
 Now, we show that the bound is sharp.
 Let $p, q$ be two distinct prime numbers. Then, every number greater than $p\cdot q - (p + q) + 1$
 could be written in the form $ap + bq$ with $a,b \ge 0$, and this is the minimal number
 with this property, i.e., $pq - (p+q)$ could not be written in that way; see~\cite{PighizziniS02}.
 Set $V_j = a_j^{p-1}(a_j^p)^*$, $W_j = a_j^{q-1}(a_j^q)^*$
 and $U_j = \{ a_j^{p + q - 2 + ap + bq} \mid a,b \ge 0 \} = V_j W_j$.
 Then
 $$
  \stc(V_1 \shuffle \ldots \shuffle V_k) = p^k.
 $$
 For, the permutation automaton $\mathcal C = (\Sigma, [p]^k, \delta, (0,\ldots,0), \{ (p-2, \ldots, p-2) \})$
 with
 $$
  \delta( (i_1, \ldots, i_{j-1}, i_j, i_{j+1},\ldots, i_k), a_j )
   = (i_1, \ldots, i_{j-1}, (i_j + 1) \bmod p, i_{j+1}, \ldots, i_k)
 $$
 accepts $V_1 \shuffle \ldots \shuffle V_k$, and it is easy to see that if a language is accepted 
 by a permutation automaton with a single final state, then this automaton is minimal.
 Similarly, $\stc(W_1 \shuffle \ldots \shuffle W_k) = q^k$.
 
 Set $n_j = pq - 1$ for $j \in \{1,\ldots, k\}$. Then the language
 \begin{multline*}
      L = U_1 \shuffle \ldots \shuffle U_k = (V_1 W_1) \shuffle \ldots \shuffle (W_1 W_k) \\ 
       = (V_1 \shuffle \ldots \shuffle V_k) \shuffle (W_1 \shuffle \ldots \shuffle W_K).
 \end{multline*}
 fulfills the prerequisites of Lemma~\ref{lem:lower_bound_sc} with $(n_1, \ldots, n_k)$,
 as for each $j \in \{1,\ldots k\}$ we have $a_j^{pq - 2} \notin U_j$ and
 $a_j^{pq-1}a_j^* \subseteq U_j$.
 Hence,
 $$
  \stc(L) = (pq)^{|\Sigma|} = (\gcd(p,q) + \lcm(p,q) - 1)^{|\Sigma|}
 $$
 and the bound given by the statement is attained. \qed 
\end{proof}

As for any two numbers $n,m > 0$
we always have $\gcd(n,m) + \lcm(n,m) - 1 \le nm$,
we can deduce the next bound in terms of the size of recognizing automata.
The result improves the general bound $(2nm)^{|\Sigma|}$ from~\cite{Hoffmann2021NISextended,DBLP:conf/cai/Hoffmann19}.
\todo{die sprache mit größer $nm$?}

\begin{theorem}
\label{thm:group:shuffle_upper_bound}
 Let $U, V \subseteq \Sigma^*$ be commutative group languages
 recognized by automata with $n$ and $m$ states.
 Then, $U\shuffle V$ is recognizable by an automaton with at most $(nm)^{|\Sigma|}$
 states.
\end{theorem}

We do not know if the last bound is sharp. The best lower bound we can give is the next one, which essentially follows
by the lower bound for concatenation in case of unary languages, see~\cite[Theorem 5.4]{YuZhuangSalomaa1994}.

\begin{propositionrep}
\label{prop:group:shuffle_lower_bound}
 Let $n, m > 0$ be coprime numbers. Then, there exist commutative group languages of states complexities
 $n$ and $m$ such that their shuffle has state complexity $nm$.
\end{propositionrep}
\begin{proof} 
 Let $a \in \Sigma$ and $n,m > 0$ be coprime numbers.
 Set $U = \{ w \in \Sigma^* \mid |w|_a \equiv n - 1 \pmod{n} \} = a^{n-1}(a^n)^* \shuffle (\Sigma\setminus \{a\})^*$
 and $V = \{ w \in \Sigma^* \mid |w|_a \equiv m - 1 \pmod{m} \} = a^{m-1}(a^m) \shuffle (\Sigma\setminus \{a\})^*$.
 Then, by a number-theoretical 
 result from~\cite[Lemma 5.1]{YuZhuangSalomaa1994},
 $U \shuffle V = F \shuffle (\Sigma\setminus \{a\})^* \cup a^{nm - 1}a^* \shuffle  (\Sigma\setminus \{a\})^*$
 for some finite $F \subseteq \Sigma^*$
 with $a^{nm - 2} \notin F$ 
 and this language has state complexity $nm$.~\qed
\end{proof}

\subsection{The Upward and Downward Closure and Interior Operations}


First, we will show that every word is contained in the downward closure
of a commutative group language.

\begin{propositionrep} 
\label{prop:group:downward_upper_bound}
 Let  $L \subseteq \Sigma^*$ be a commutative group language.
 Then the downward closure $\mathop{\downarrow\! L}$ equals $\Sigma^*$.
\end{propositionrep}
\begin{proof}
 Let $\mathcal A = (\Sigma, Q, \delta, q_0, F)$
 be a permutation automaton with $L(\mathcal A) = L$.
 If $L = \{\varepsilon\}$, then the statement is clear.
 So, suppose $L \ne \{\varepsilon\}$. 
 Note that $L \ne \emptyset$, as $\emptyset$ is never a group language.
 Let, without loss of generality,
 $\{ a \mid \exists w \in L : |w|_a > 0 \}^* = \{ a_1,\ldots,a_l\}$, $0 < l \le k$.
 Choose, for any $j \in \{1,\ldots,l\}$,
 a word $w_j \in L$ with $|w_j|_{a_j} > 0$.
 Set $w = w_1 \cdots w_l$.
 Then, some power $w^m$ acts as the identity on $Q$.
 Hence, $w_1 w^{r \cdot m} \in L$ for any $r \ge 0$. 
 Let $u \in \{a_1,\ldots,a_l\}^*$ be arbitrary, then,
 as 
 \[
  |w_1 w^{r\cdot m}|_{a_j} < |w_1 w^{(r+1) \cdot m}|_{a_j}
 \]
 for any $r \ge 0$ and $j \in \{1,\ldots,l\}$, there exists
 $r \ge 0$
 such that, for any $j \in \{1,\ldots,l\}$,
 \[
  |u|_{a_j} < |w_1 w^{r\cdot m}|_{a_j}.
 \]
 But then, we can reorder the letters in $w_1 w^{r\cdot m}$
 such that $u$ is a scattered subsequence of this word, i.e., we can delete letters
 from this word such that $u$ results.
 More formally, we have $(u \shuffle \Sigma^*) \cap \perm(w_1 w^{r\cdot m}) \ne \emptyset$.
 As $L$ is commutative, $\perm(w_1 w^{r\cdot m}) \subseteq L$
 and we find $u \in \downarrow\! L$.~\qed 
\end{proof}

\todo{beweisen, mit satz aus anhang?}
This is not true for general commutative languages,
see Proposition~\ref{thm:aperiodic:interior_lower_bound}.

\begin{proposition}
\label{prop:group:upward_closure_lower_bound}
 Let $n > 0$. There exists a commutative group language $L \subseteq \Sigma^*$
 with period vector $(n,1,\ldots,1)$
 such that $\stc(L) = n$
 and its upward closure has state complexity $n$
 with index vector $(n-1,1,\ldots,1)$
 and period vector $(1,\ldots,1)$.
\end{proposition}
\begin{proof} 
 Let $a \in \Sigma$ and $n > 0$.
 Set $L = \{ w \in \Sigma^* \mid |w|_a \equiv n - 1 \pmod{n} \} = a^{n-1}(a^n)^* \shuffle (\Sigma\setminus \{a\})^*$.
 Then, $\upwardclosure{L} = a^{n-1}a^* \shuffle (\Sigma\setminus \{a\})^*$
 and $\stc(\upwardclosure{L}) = n$.\qed
\end{proof}

\subsection{Union and Intersection}

\begin{theorem}
\label{thm:group:union_intersection}
 For any alphabet $\Sigma$
 and commutative group language of state complexities
 $n$ and $m$, the intersection and union
 is recognizable by an automaton with $nm$
 states. Furthermore, there exists commutative group
 languages with state complexities such that
 every automaton for their union (intersection)
 needs $nm$ states.
\end{theorem}
\begin{proof}
 The upper bound holds for regular languages in general~\cite{YuZhuangSalomaa1994}. 
 Let $a \in \Sigma$ and $n,m > 0$ two coprime numbers. 
 The languages for the lower bound are similar
 to the ones given in~\cite{YuZhuangSalomaa1994}, but
 also work for $\Sigma = \{a\}$ (in~\cite{YuZhuangSalomaa1994} $|\Sigma|\ge 2$ is assumed).
 Set $U = \{ w \in \Sigma^* \mid |w|_a \equiv 0 \pmod{n} \} = (a^n)^* \shuffle (\Sigma\setminus \{a\})^*$
 and $V = \{ w \in \Sigma^* \mid |w|_a \equiv 0 \pmod{m} \}$.
 Then, $\stc(U \cap V) = \stc(U \cup V) = nm$.\qed  
\end{proof}

\section{The Case of Commutative Aperiodic Languages}

For shuffle (Theorem~\ref{thm::sc_aperiodic}), union and intersection (see~\cite[Theorem 4.3]{YuZhuangSalomaa1994}), we already have upper bounds, see Table~\ref{tab:sc_results_here}.
Note that the lower bound construction for Boolean operations
given in~\cite[Theorem 1 \& 8]{BrzozowskiL12} over an at least
binary alphabet uses
the commutative and aperiodic languages $\{ w \in \Sigma^* \mid |w|_a = n - 2 \}$
and $\{ w \in \Sigma^* \mid |w|_b = m - 2 \}$ for $n,m \ge 2$. For the upward closure, Proposition~\ref{prop:aperiodic:upward_lower_bound}
gives a lower bound. So, here, we handle the missing
cases of the downward closure and the shuffle operation.

But first, we give a similar normal form theorem for aperiodic
commutative languages as for group languages.
The proof is essentially the same as in the group case.

\begin{theorem}
 Let $\Sigma = \{a_1, \ldots, a_k\}$
 and $L \subseteq \Sigma^*$.
 Then, the following conditions are equivalent:
 \begin{enumerate}
 \item $L$ is a commutative and aperiodic language;
 \item $L$ is a finite union of languages of the form 
  $
   U_1 \shuffle \ldots \shuffle U_k
  $
  with $U_j \subseteq \{ a_j \}^*$ aperiodic
  and recognizable by an automaton with a single final state;
 \item $L$ is a finite union of languages of the form 
  $
   U_1 \shuffle \ldots \shuffle U_k
  $
  with $U_j \subseteq \{ a_j \}^*$ aperiodic.
 \end{enumerate}
\end{theorem}


Next, we state a lower bound for the downward closure. By Equation~\eqref{eqn:interior}, using the complemented language,
this implies the same lower bound for the upward interior.

\begin{proposition}
\label{thm:aperiodic:interior_lower_bound} 
 Let $a \in \Sigma$ and $n > 0$. Set $L = \{a^n\}$.
 Then,
 $\mathop{\downarrow\!} L = \{ \varepsilon, a, \ldots, a^n \}$
 and so $\stc(\mathop{\downarrow\!} L) = \stc(L) = n + 2$.
\end{proposition}

For concatenation it is known that the state complexity 
for unary languages is $nm$~\cite[Theorem 5.4 \& Theorem 5.5]{YuZhuangSalomaa1994}. In fact, the witness languages for the lower bound are group languages (see also Prop~\ref{prop:group:shuffle_lower_bound}).
For unary aperiodic languages, concatenation has state complexity $n + m - 1$~\cite[Theorem 8]{BrzozowskiL12}. By using these unary witness languages and introducing self-loops
for additional letters, and as in the unary case (and for the mentioned extension to more letters) concatenation
and shuffle coincide, this immediately gives lower bounds for the shuffle operation on commutative languages as well.

However, in the next result, we show that we can do better for aperiodic (even for finite) commutative languages.

\begin{proposition}
\label{prop:aperiodic:shuffle_lower_bound}
 Let $\Sigma$ be an at least binary alphabet. 
 Then for each even $n, m > 0$
 there exist commutative and finite languages $U, V \subseteq \Sigma^*$ with $\stc(U) = n$, $\stc(V) = m$
 such that $\stc(U \shuffle V) \ge \frac{nm}{4} + 1$.
\end{proposition}
\begin{proof}
 We give the construction for a binary alphabet, as for larger alphabets
 the lower bound is implied by adding self-loops to the automaton for each additional letter.
 Let $\Sigma = \{a,b\}$ and $N, M \ge 0$
 with  $n = |\Sigma|N + 2$
 and $m = |\Sigma|M + 2$.
 Set $U = \{a^N, b^N\}$ and  $V = \{a^M, b^M\}$.
 Then $\stc(U) = n$ and $\stc(V) = m$.
 Set \todo{bild in errata/comments-datei?}
 \begin{multline*}
      L = U\shuffle V  = \{ u \in \Sigma^* \mid (|u|_a = N + M, |u|_b = 0) \lor (|u|_a = N, |u|_b = M) \\ 
 \lor (|u|_a = M, |u|_b = N)
 \lor (|u|_a = 0, |u|_b = N + M) \}.
 \end{multline*} 
 We show that the words $a^{i} b^{j}$
 for $i \in \{0,1,\ldots,N\}$
 and $j \in \{0,1,\ldots,M\}$
 are pairwise inequivalent for the Nerode right-congruence.
 Let $(n_1, n_2), (m_1, m_2) \in \{0,1,\ldots,N\} \times \{0,1,\ldots,M\}$
 with $(n_1, n_2) \ne (m_1, m_2)$.
 First, suppose $n_i \ne 0$ for $i \in \{0,1\}$
 and $m_i \ne 0$ for $i \in \{0,1\}$.
 Set $u = a^{N - n_1} b^{M - n_2}$.
 Then $a^{n_1} b^{n_2} u \in L$.
 By the assumptions, $a^{m_1} b^{m_2} u \in L$
 if and only if $(m_1, m_2) + (N - n_1, M - n_2) = (M, N)$.
 So, if this is not the case, then $u$
 distinguishes both words.
 Otherwise,
 we must have $(n_1, n_2) + (N- m_1, M - m_2) \ne (M, N)$,
 for if not then $M - N = n_1 - m_1 = m_1 - n_1$.
 The last equality then gives $n_1 = m_1$
 and similarly we find $n_2 = m_2$. So, in that case
 with $v = a^{N - m_1} b^{M - m_2}$
 we have $a^{m_1} b^{m_2} v \in L$
 but $a^{n_1} b^{n_2} v \notin L$.
 Lastly, if at least one component is zero,
 we must distinguish more cases
 as
 $\{ a^{N + M}, b^{N + M} \}\subseteq L$,
 which are left out here due to
 space. 
 
 So, we have $\stc(L) \ge (N+1)(M+1) + 1$ (the additional one accounts for a trap state necessary, as we measure the state complexity in terms of complete automata).
 Hence, $\stc(L) \ge \frac{nm}{4} + 1$. \qed
\end{proof}

\section{Conclusion} As shown in  Table~\ref{tab:sc_known_results} and  Table~\ref{tab:sc_results_here},
 for shuffle and the upward and downward closure and interior operations,
 it is not known if the given upper bounds are tight.  Hence, this is an open problem.

 {\smallskip \noindent \footnotesize
 \textbf{Acknowledgement.} I thank the anonymous referees of~\cite{Hoffmann2021NISextended} (the extended version of~\cite{DBLP:conf/cai/Hoffmann19}) whose feedback also helped in the present work.
 I also thank the referees of the present work
 for critical and careful reading, and pointing
 out typos and parts that needed better explanation.
 
 }
\bibliographystyle{splncs04}
\bibliography{ms,perm} 
\end{document}